\documentclass[12pt]{iopart}
\usepackage[utf8]{inputenc}
\usepackage[T1]{fontenc}
\usepackage{eucal}
\usepackage{microtype}
\usepackage{csquotes}
\usepackage{graphicx} 
\expandafter\let\csname equation*\endcsname\relax
\expandafter\let\csname endequation*\endcsname\relax
\usepackage{amsthm}
\usepackage{mathtools}
\usepackage{amssymb,amsfonts,stmaryrd}
\usepackage{mathrsfs}  
\usepackage{tikz-cd}
\usepackage[normalem]{ulem}

\mathtoolsset{showonlyrefs}

\newcommand*{\RR}{\mathbb{R}}

\newcommand*{\dd}{\mathrm{d}}

\newcommand*{\contr}[1]{\iota_{#1}}
\newcommand*{\liedv}[1]{\mathcal{L}_{#1}}
\newcommand*{\Reeb}{\mathcal{R}}
\DeclareMathOperator{\pr}{pr}

\let\Im\Image
\DeclareMathOperator{\spn}{span}

\DeclareMathOperator{\grad}{grad}

\newcommand{\sode}{\Gamma}

\newcommand{\prtr}{\pr_{TQ\times \RR}}
\newcommand{\qtilde}{\tilde{\mathcal {Q}}}
\newcommand{\ptilde}{\tilde{\mathcal {P}}}

\let\hat\widehat

\theoremstyle{plain}
\newtheorem{theorem}{Theorem}
\newtheorem*{theorem*}{Theorem}

\newtheorem*{lemma*}{Lemma}
\newtheorem{proposition}[theorem]{Proposition}
\newtheorem*{proposition*}{Proposition}
  
\newtheorem*{corollary*}{Corollary}
\theoremstyle{definition}

\newtheorem*{definition*}{Definition}
\newtheorem{example}{Example}
\newtheorem*{example*}{Example}
\theoremstyle{remark}

\newtheorem*{remark*}{Remark}

\newtheorem*{conjecture*}{Conjecture}

\newtheorem*{problem*}{Problem}

\usepackage[citestyle=numeric-comp,sorting=nyt,sortcites, backend=biber, giveninits=true, maxnames=99]{biblatex}
\usepackage[unicode=true, pdfusetitle, colorlinks=true, allcolors=blue, bookmarks=false]{hyperref}
\usepackage[colorinlistoftodos]{todonotes} 

\bibliography{biblio}

\AtEveryBibitem{
  \ifentrytype{online}{}{
    \ifentrytype{thesis}{}{
        \clearfield{url}
      }
  }
}

\DeclareSourcemap{
  \maps[datatype=bibtex]{
    \map[overwrite=true]{
      \step[fieldset=urldate, null]
      \step[fieldset=language, null]
      \step[fieldset=address, null]
     \step[fieldset=pagetotal, null]
    }
  }
}

\begin{document}

\title[Contact Lagrangian systems subject to impulsive constraints]{Contact Lagrangian systems subject to impulsive constraints}

\author{Leonardo Colombo$^{1}$, Manuel de Le\'on$^{2,3}$, Asier L\'opez-Gord\'on$^{2}$}

\date{\today}

\address{$^{1}$ Centro de Automática y Robótica (CSIC-UPM), 
Carretera de Campo Real, km 0, 200, 28500 Arganda del Rey, Spain.\\$^{2}$Instituto de Ciencias Matemáticas (CSIC-UAM-UC3M-UCM) 
Calle Nicolás Cabrera, 13-15, Campus Cantoblanco, UAM, 28049 Madrid, Spain.\\ $^{3}$Real Academia de Ciencias Exactas, Físicas y Naturales Calle Valverde, 22, 28004, Madrid, Spain.}
\ead{leonardo.colombo@car.upm-csic.es, mdeleon@icmat.es, asier.lopez@icmat.es}

\vspace{10pt}

\begin{abstract}
We describe geometrically contact Lagrangian systems under impulsive forces and constraints, as well as instantaneous nonholonomic constraints which are not uniform along the configuration space. In both situations, the vector field describing the dynamics of a contact Lagrangian system is determined by defining projectors to evaluate the constraints by using a Riemannian metric. In particular, we introduce the Herglotz equations for contact Lagrangian systems subject to instantaneous nonholonomic constraints. Moreover, we provide a Carnot-type theorem for contact Lagrangian systems subject to impulsive forces and constraints, which characterizes the changes of energy due to contact-type dissipation and impulsive forces. We illustrate the applicability of the method with practical examples, in particular, a rolling cylinder on a springily {plane} and a rolling sphere on a non-uniform {plane}, both with dissipation. 
\end{abstract}

%
%
%
%
%

\section{Introduction}
In recent years, there has been an increasing interest in the study of contact mechanical systems (see \cite{Rivas2022,deLeon2019g, deLeon2021h,Gaset2020,Lainz2022,deLeon2021l,Bravetti2017c,Bravetti2017d} and references therein). Contact geometry has been used in the last years to describe dissipative mechanical systems, as well as systems in thermodynamics \cite{Simoes2020a,Gay-Balmaz2018,Mrugala2000a,Eberard2007a}, quantum mechanics \cite{Ciaglia2018b}, control theory \cite{deLeon2020k}, dissipative field theories \cite{Gaset2020c,Gaset2021,Rivas2022}, etc. In the contact Lagrangian formalism, the equations of motion are obtained using the Herglotz variational principle instead of Hamilton's principle of least action, so that these dynamical systems do not enjoy conservative properties, but dissipative ones \cite{Herglotz1930b,deLeon2019f}. The main difference between both variational principles is that in the Herglotz variational principle the action is defined by a non-autonomous ODE instead of an integral. Contact Hamiltonian and Lagrangian systems with nonholonomic constraints were introduced by de Le\'on, Jim\'enez and Lainz in \cite{deLeon2021h}.

A geometrical formulation for mechanical systems with one-sided constraints was developed by Lacomba and Tulczyjew \cite{Lacomba1990}.
Ibort \textit{et al.} studied the geometrical aspects of Lagrangian systems subject to impulsive and one-sided constraints in a series of papers \cite{Ibort1998a,Ibort2001a,Ibort1997b}. This was extended to the Hamiltonian formalism by Cortés and Vinogradov \cite{Cortes2006a}. Additionally, Cortés \textit{et al.} studied Lagrangian systems subject to generalized non-holonomic constraints \cite{Cortes2001b,Cortes2001a,Cortes2006a}. The aim of this paper is to go one step further on contact Lagrangian systems with constraints and consider impulsive forces and constraints, as well as instantaneous nonholonomic constraints which are not uniform along the configuration space.  We develop a Carnot-type theorem for contact Lagrangian
systems subject to impulsive forces and constraints, which characterizes the changes
of energy due to contact-type dissipation and impulsive forces. In particular, the results by Ibort \textit{et al.} and Cortés \textit{et al.} are recovered from our formalism in the limit where there is no dissipation (i.e., the contact Lagrangian is an usual Lagrangian and the Herglotz equations yield the classical Euler-Lagrange equations).

The remainder of the paper is structured as follows. In Section \ref{section_review} we review constrained contact Lagrangian systems. Contact Lagrangian systems subject to impulsive forces are introduced in Section \ref{section_impulsive}. Contact Lagrangian systems with instantaneous nonholonomic constraints are introduced in Section \ref{section_instantaneous}. Finally, some conclusions and related topics for future research are given in Section \ref{section_conclusions}.

\section{Contact Lagrangian systems subject to constraints}
\label{section_review}
{
In this section, we recall the main properties of contact Hamiltonian and Lagrangian systems. See \cite{deLeon2019g,deLeon2019f,Rivas2022,deLeon2021h,Lainz2022,Gaset2020c,Bravetti2017c,Bravetti2017d} for more details. We also review contact Lagrangian systems subject to nonholonomic constraints as developed in \cite{deLeon2021h}.
}

A \emph{contact manifold} is a pair $(M, \eta)$, where $M$ is an $(2 n+1)$-dimensional differentiable manifold and $\eta$ is a $1$-form on $M$ called \emph{contact form} such that $\eta \wedge(\mathrm{d} \eta)^{n}$ is a volume form. Given a contact manifold $(M,\eta)$, there exists a unique vector field $\Reeb$ on $M$ such that $\contr{\Reeb} \dd \eta = 0$ and $\contr{\Reeb} \eta = 1$. $\Reeb$ is called \emph{Reeb vector field}. In \emph{Darboux} (local) \emph{coordinates} $(q^i, p_i, z)$, the contact form is written as $\eta = \dd z - p_i \dd q^i$, and the Reeb vector field as $\Reeb = \frac{\partial  } {\partial z}$.

The contact structure $\eta$ on $M$ defines the so-called \emph{musical isomorphisms} $\flat : TM \to T^*M$ as the map $v \mapsto \contr{v} \dd \eta + \eta(v) \eta$, and $\sharp = \flat^{-1}$. 
They induce the isomorphisms of $C^\infty(M)$-modules $\flat: \mathfrak{X}(M)\to \Omega^1(M)$ and $\sharp: \Omega^1(M) \to \mathfrak{X}(M)$.
Given a Hamiltonian function $H$ on $(M,\eta)$, we define the \emph{(contact) Hamiltonian vector field} $X_H$ by
\begin{equation}
  \flat(X_H) = \dd H  - \left(\Reeb (H) + H  \right) \eta.
  \label{Hamiltonian_vector_field}
\end{equation}
The triple $(M,\eta, H)$ is called a \emph{contact Hamiltonian system}.
Eq.~\eqref{Hamiltonian_vector_field} is equivalent to 

\begin{equation}
  \eta(X_H) = -H,\qquad \liedv{X_H} \eta = - \Reeb(H) \eta,
\label{Eqs_Hamiltonian_vector_field}
\end{equation}
where $\liedv{X_H} \eta$ is the Lie derivative of $\eta$ with respect to $X_H$.
Additionally, the following identities hold:
\begin{equation}
  X_H(H) = -\Reeb(H)\ H,\,\qquad\contr{X_H} \dd \eta = \dd H- \Reeb(H) \eta.
\end{equation} 
{
The equations with the Lie derivative mean that
\begin{equation}
  \frac{\mathrm{d} } {\mathrm{d}t} \left(H \circ \varphi_t \right) = -\left(\Reeb(H)\circ \varphi_t\right) \left( H\circ \varphi_t\right),
  \qquad \frac{\mathrm{d} } {\mathrm{d}t} \varphi_t^\ast \eta = -\left(\Reeb(H)\circ \varphi_t\right) \varphi_t^\ast \eta,
\end{equation}
where $\varphi_t$ is the flow of $X_H$.}
Integrating these equations yields
\begin{equation}
   H \circ \varphi_t = \exp \left(\int_0^t -\Reeb(H)(\varphi_\tau) \dd \tau   \right) H,\,\,\quad \varphi_t^* \eta = \exp \left(\int_0^t -\Reeb(H)(\varphi_\tau) \dd \tau   \right) \eta.
   \label{Eqs_Hamiltonian_flow}
\end{equation}

 In Darboux coordinates $(q^i,p_i,z)$ the flow of $X_H$ is given by $$\dot{q}^i =\frac{\partial H}{\partial p_i},\,\,\,\dot{p}_i=-\frac{\partial H}{\partial q^i}-p\frac{\partial H}{\partial z},\,\,\,\dot{z}=p
 _i\frac{\partial H}{\partial p_i}-H.$$

Next, consider a Lagrangian function $L:TQ\times \RR \to \RR$ and let us introduce the 1-form $\alpha_L=S^*(\dd L)$, where $S^*$ is the adjoint operator of the vertical endomorphism on $TQ$ extended in the natural way to $TQ\times \RR$. That is, if $(q^i)$ denote local coordinates on $Q$ and $(q^i, \dot q^i, z)$ the induced coordinates on $TQ\times \RR$, locally, $S = \dd q^i \otimes \frac{\partial  } {\partial  \dot{q}^i}$, and then $\alpha_L =S^*(\dd L)= \frac{\partial L} {\partial \dot q^i} \dd q^i$.

Let $\eta_L$ be the 1-form on $TQ\times \RR$ given by $\eta_L = \dd z - \alpha_L = \dd z - \frac{\partial L} {\partial \dot q^i} \dd q^i$, one can show that $\eta_L$ is a contact form if and only if $L$ is \emph{regular}, i.e., the Hessian matrix $(W_{ij}) = \left(\frac{\partial ^2 L} {\partial \dot q^i \dot q^j}  \right)$ is regular. Hereinafter, we shall assume that $L$ is regular.

The \emph{energy} of the system is given by $E_L = \Delta(L) - L$, where $\Delta = \dot q^i \partial / \partial \dot q^i$ is the Liouville vector field on $TQ$ trivially extended to $TQ\times \RR$. The triple $(TQ\times \RR, \eta_L, E_L)$ is called \emph{contact Lagrangian system}, and its corresponding Reeb vector field $\Reeb_L$ is locally given by 
\begin{equation}
  \Reeb_L = \frac{\partial  } {\partial z} - W^{ij} \frac{\partial ^2 L } {\partial \dot q^i \partial z} \frac{\partial  } {\partial \dot q^j}, 
\end{equation}
where $(W^{ij})=(W_{ij})^{-1}$. The dynamics of the contact Lagrangian system is given by the \emph{Lagrangian vector field} $\sode_L$, given by
\begin{equation}
  \flat_L (\sode_L) = \dd E_L - \left(E_L + \Reeb_L(E_L)  \right) \eta_L,
  \label{SODE_energy}
\end{equation}
where $\flat_L$ denotes the musical isomorphism defined by the contact form $\eta_L$. Eqs.~\eqref{Eqs_Hamiltonian_vector_field} can be now written as 
\begin{equation}
\eta_L(\sode_L) = -E_L,\quad\liedv{\sode_L} \eta_L = - \Reeb_L(E_L) \eta_L.
\label{Eqs_Lagrangian_SODE}
\end{equation}
Moreover, Eqs.~\eqref{Eqs_Hamiltonian_flow} are now written as
\begin{equation}
E_L \circ \tilde{\varphi}_t = \exp \left(\int_0^t -\Reeb_L(E_L)(\tilde{\varphi}_\tau) \dd \tau   \right) E_L,\quad
 \tilde{\varphi}_t^* \eta_L = \exp \left(\int_0^t -\Reeb_L(E_L)(\tilde{\varphi}_\tau) \dd \tau   \right) \eta_L, 
\label{Eq_Lagrangian_flow_energy}
\end{equation}
where $\tilde{\varphi}_t$ is the flow of $\sode_L$.

Let us recall that a vector field $\sode$ on $TQ\times \RR$ is called a \emph{second order differential equation} (\emph{SODE}) if $S(\sode) = \Delta$. Locally, a SODE is of the form
\begin{equation}
  \sode = \dot q^i \frac{\partial  } {\partial q^i} + \sode^i(q, \dot q, z) \frac{\partial  } {\partial \dot q^i} + \sode^z(q, \dot q, z) \frac{\partial } {\partial z},\,i=1,\ldots,n.
\end{equation}

A vector field $\sode$ on $TQ\times \RR$ is a SODE if and only if any integral curve of $\sode$ can be locally written as $(\sigma(t), \dot \sigma(t), z(t))$ for some local curve $\sigma$ on $Q$ and some local curve $z$ on $\RR$. This curve $\sigma$ is called a \emph{solution} of the SODE $\sode$. Since the Lagrangian $L$ is regular, one has the following equivalence between the \textit{Herglotz equations} and the geometric dynamical equations \eqref{SODE_energy} (see also ~\cite{deLeon2019f}).

\begin{proposition}
Let $L$ be a regular contact Lagrangian system on $TQ\times \RR$, and let $\sode_L$ be the Hamiltonian vector field associated with the energy given by Eq.~\eqref{SODE_energy}. Then $\sode_L$ is a SODE on $TQ\times \RR$. Moreover, the integral curves of $\sode_{L}$ are solutions of the \textit{Herglotz equations} \begin{equation}
    \frac{\partial L}{\partial q^{i}}-\frac{\mathrm{d}}{\mathrm{d} t} \frac{\partial L}{\partial \dot{q}^{i}}+\frac{\partial L}{\partial \dot{q}^{i}} \frac{\partial L}{\partial z}=0.
    \label{Herglotz_eqs}
\end{equation}
\end{proposition}

In the particular case the Lagrangian function is given by
\begin{equation}\label{natural}
  L = \frac{1}{2}g_{ij} \dot q^i \dot q^j - V(q,z),
\end{equation}
where $g$ is a (pseudo)Riemannian metric on $Q$, then the Herglotz equations yield
\begin{equation}\label{eqmetric}
{
  \ddot q^i + \Gamma^i_{\ jk} \dot q^j \dot q^k + g^{ij} \frac{\partial V} {\partial q^j} + \dot q^i \frac{\partial V} {\partial z} = 0,
  }
\end{equation}
where $\Gamma^i_{\ jk}$ are the Christoffel symbols of the Levi-Civita connection $\nabla$ determined by $g$ and $g^{ij}$ are the components of the inverse matrix associated with $g$. In other words, a curve $\sigma$ on $Q$ is a solution of $\sode_L$ if and only if
\begin{equation}
  \nabla_{\dot \sigma (t)} \dot \sigma(t) = -\grad V(\sigma (t)) - \frac{\partial V} {\partial z} (\sigma(t))\ \dot \sigma(t),
\end{equation}
where $\grad$ denotes the gradient with respect to $g$. Note that in the absence of potential, equation \eqref{eqmetric} are just the geodesic equations associated with the Levi-Civita connection.

Next, consider the contact Lagrangian system $L$ on $TQ\times \RR$ is restricted to certain (linear) constraints on the velocities modelled by a regular distribution $\mathcal{D}$ on the configuration manifold $Q$ of codimension $k$. Then, $\mathcal{D}$ may be locally described in terms of independent linear constraint functions $\left\{\Phi^{a}\right\}_{a=1, \ldots, k}$ by $\mathcal{D}=\left\{v \in T Q \mid \Phi^{a}(v)=0\right\}$, where $\Phi^{a}=\Phi_{i}^{a}(q) \dot{q}^{i}$. With a slight abuse of notation, we shall also denote by $\Phi^a$ the associated 1-forms on $Q$. 
{More generally, one could consider constraints $\Phi_a$ also depending on $z$. Nevertheless, as far as we know, there are no physical examples of constraints which depend on the variable $z$.}

{{
So, a curve $\sigma$ on $TQ\times \RR$ satisfies the Herglotz variational principle with constraints if and only if it satisfies the \emph{constrained Herglotz equations} \cite{deLeon2021h}, namely}}
{The dynamics $\sigma(t)$ of a contact Lagrangian system subject to nonholonomic constraints are given by the \emph{constrained Herglotz equations}:}
  \begin{equation} \label{Herglotz_eq_constrained}
  \frac{\mathrm{d}}{\mathrm{d} t} \frac{\partial L}{\partial \dot{q}^{i}}-\frac{\partial L}{\partial q^{i}}-\frac{\partial L}{\partial \dot{q}^{i}} \frac{\partial L}{\partial z}=\lambda_{a} \Phi_{i}^{a},
   \quad 
  \Phi^{a}(\dot{\sigma}(t))=0,\, a=1,\ldots,k,
\end{equation}
{for} some functions $\lambda_a$. These equations can be obtained variationally from the Herglotz principle with constraints (see \cite{deLeon2021f}).

Note that if the Lagrangian is given by \eqref{natural} then the constrained Herglotz equations are
\begin{equation}
{
  \ddot q^i + \Gamma^i_{\ jk} \dot q^j \dot q^k + g^{ij} \frac{\partial V} {\partial q^j} + \dot q^i \frac{\partial V} {\partial z} = g^{ij}\lambda_a \Phi_j^a.
  }
\end{equation}
 In other words, a curve $\sigma$ on $Q$ satisfies the Herglotz variational principle if and only if
\begin{subequations}
\begin{flalign}
  &  \nabla_{\dot \sigma (t)} \dot \sigma(t) = -\grad V(\sigma (t)) - \frac{\partial V} {\partial z} (\sigma(t))\ \dot \sigma(t)
  + \lambda (\sigma(t)),\\
  & \dot \sigma(t) \in \mathcal D_{\sigma(t)},
\end{flalign}
\end{subequations}
where $\lambda$ is a section of $\mathcal D^\perp$ along $\sigma$, and $\mathcal D^\perp$ denotes the orthogonal complement of $\mathcal D$ with respect to the metric $g$.


\section{Contact Lagrangian systems subject to impulsive forces and constraints}
\label{section_impulsive}

{
In this section, we extend the notion of impulsive constraints (see \cite{Brogliato1996,Cortes2006a,Ibort2001a,Ibort1997b,Ibort1998a,Lacomba1990,Rosenberg1977}) to contact Lagrangian systems.
}

Consider a system of $n$ particles in $\RR^3$ such that the $j$-th particle has mass $m_j$. Let us introduce the coordinates $(q^{3j-2}, q^{3j-1}, q^{3j})$ for the $j$-th particle. Suppose that $F_j = (F^{3j-2}, F^{3j-1}, F^{3j})$ is the net force acting on the particle $j$. The equations of motion of the $j$-th particle in the interval $[t_0, t_1]\subset \RR$ are then given by
\begin{equation}
   \dot q^{k}(t_1) = \frac{1}{m_j} \int_{t_0}^{t_1} F^k(\tau)\ \dd \tau + \dot q^k (t_0), 
 \end{equation} 
 where $3j-2\leq k \leq 3j$. This equation is a generalization of the classical Newton's second law, since it allows one to consider the case of finite jump discontinuities \cite{Rosenberg1977}. This is the case of impulsive forces, which produce a non-zero impulse at some time instant. More {precisely}, if $F$ is impulsive then
 \begin{equation}
    \lim_{t\to t_0^+} \int_{t_0}^t F(t)\ \dd \tau = P \neq 0,
  \end{equation} 
for some instant $t_0$. This implies that the impulsive force has an infinite magnitude at $t_0$, but $P$ is well-defined and bounded. This can be expressed as
\begin{equation}
  \lim_{t\to t_0^+} F(t) = P\ \delta (t_0),
\end{equation}
where $\delta$ denotes the Dirac delta. The impulsive forces may be caused by constraints, the so-called \emph{impulsive constraints}. Non-holonomic constraints of the form $\psi=0$ for $\psi = \psi_i(q) \dot q^i$ have an associated constraint force given by $F_k = \mu\ \psi_k$, where $\mu$ is a Lagrange multiplier. Thus the constraint is impulsive if
\begin{equation}
  \lim_{t\to t_0^+} \int_{t_0}^t \mu \ \psi_k\  \dd \tau = P_k \neq 0.
\end{equation}


The impulsive force may be due to a discontinuity at $t_0$ of $\psi_k$ of $\mu$ or of both. Let us assume that the constraints $\psi$ are smooth and hence the impulsive force is caused by a discontinuity of the Lagrange multiplier.

Let us now consider a contact Lagrangian system subject to a set of impulsive linear non-holonomic constraints $\left\{\psi^a  \right\}_{a=1}^r$.
From Eq.~\eqref{Herglotz_eq_constrained} we have that
\begin{equation}\label{constrained_Herglotz_form}
  \frac{\mathrm{d}p_i} {\mathrm{d}t} = \frac{\partial L}{\partial q^{i}} + p_i \frac{\partial L}{\partial z} + \lambda_{a} \psi_{i}^{a},
\end{equation}
where $p_i = \partial L/ \partial \dot q^i$.
Then
\begin{equation}
  \lim_{t\to t_0^+} \int_{t_0}^t \frac{\mathrm{d}p_i} {\mathrm{d}t}\ 
  \delta q^i \ \dd t
   = \lim_{t\to t_0^+} \int_{t_0}^t  \left[\left(\frac{\partial L}{\partial q^{i}} + p_i \frac{\partial L}{\partial z} + \lambda_{a} \psi_{i}^{a}\right) \delta q^i\right] \dd t,
\end{equation}
where $\delta q^i$ are virtual displacements satisfying the constraints, i.e., 
\begin{equation}
  \psi_i^a \ \delta q^i = 0.\label{impulsive_constraints}
\end{equation}

Since $\delta q^i$ does not depend on time and $\partial L/\partial q^i$, $p_i$ and $\partial L/\partial z$ are bounded, the first two terms of the integral on the right-hand side vanish. The third term vanishes as well by Eq.~\eqref{impulsive_constraints}. Therefore we have that $\left[p_i(t_0^+) - p_i (t_0)  \right] \delta q^i = 0$. In other words, the change of momentum $\Delta p_i$ satisfies 
\begin{equation}
  \Delta p_i\ \delta q^i = 0,
\end{equation}
where $\delta q^i$ is constrained by the condition \eqref{impulsive_constraints}, and thus
\begin{equation}
  \Delta p_i = \bar \mu_a \psi_i^a,
\end{equation}
where $\bar \mu_a$ are some Lagrange multipliers. 

If the Lagrangian is given by \eqref{natural} then
\begin{equation}
{
  \Delta \dot q^i \coloneqq \dot q^i (t_0^+) - \dot q^i (t_0^-) 
  = g^{ij} \bar \mu_a \psi_j^a. }
\end{equation}

\subsection{Holonomic one-sided constraints}

Consider a contact Lagrangian system $L$ on $TQ\times \RR$ subject to a holonomic one-sided constraint $\Psi(q)\geq 0$ (e.g., the collision with a fixed wall). This inequality determines a closed subset of $Q$, whose boundary $N$ is a $(n-1)$-dimensional submanifold of $Q$. Suppose that the Lagrangian function $L$ is given by \eqref{natural}. Then, the equations of motion are
\begin{equation}
\begin{array}{ll}
     \nabla_{\dot \sigma (t)} \dot \sigma(t) = -\grad V(\sigma (t)) - \dfrac{\partial V} {\partial z} (\sigma(t))\ \dot \sigma(t),
    & \text{if } \Psi(\sigma(t))>0,\\ \\
    \Delta \dot \sigma(t) = \dot \sigma(t^+) - \dot \sigma(t^-) 
    \in T_{\sigma(t)}^\perp N ,
     & \text{if } \Psi(\sigma(t))=0,
\end{array}
\end{equation}
where $T_{\sigma(t)}^\perp N$ is the orthogonal complement of $ T_{\sigma(t)} N$ with respect to the metric $g$. 
In other words, if $\sigma(t)=(q^i(t))$ {with $\psi(\sigma(t))=0$} we have
\begin{equation}
  \dot q^i (t^+) - \dot q^i (t^-) = \bar \mu g^{ij} \frac{\partial \Psi } {\partial q^j},
\end{equation}
that is,
\begin{equation}
  \Delta \dot \sigma(t) = \bar \mu \grad \Psi.
\end{equation}
Let us introduce the orthogonal projectors  $ \qtilde: TQ \to T^\perp N$, $\ptilde: TQ \to TN$.  

Locally,
\begin{equation}
  \qtilde (X) =  \frac{ g(\grad \Psi, X)}{g (\grad \Psi, \grad \Psi)} \grad \Psi \label{projector_Q_local}
\end{equation}
for any vector field $X$ on $Q$.

Suppose that the normal components of the velocities before and after the impact are related by $\dot \sigma(t^+)^\perp = -\alpha \dot \sigma(t^-)^\perp$, where $\alpha$ is the restitution coefficient. In other words, $\dd \Psi(\dot \sigma(t^+)) = - \alpha  \dd \Psi(\dot \sigma(t^-))$, or, equivalently, $\qtilde (\dot \sigma(t^+)) = - \alpha \qtilde (\dot \sigma(t^-))$. Note that for $\alpha=1$ we have elastic collisions, without energy loss, while for $0<\alpha<1$ we have plastic collisions which in general have energy loss. In the case $\alpha=0$ we have completely inelastic collisions.

On the other hand,
\begin{equation}
  \ptilde (\dot \sigma(t^+)) =  \ptilde (\dot \sigma(t^-)),
\end{equation}
and hence
\begin{equation}
  \dot \sigma(t^+) = \left(\ptilde - \alpha \qtilde  \right) (\dot \sigma(t^-)). \label{change_velocity_projector}
\end{equation}

 Let $V$ be a real vector space endowed with an inner product $\langle\cdot ,\cdot \rangle$ . Suppose that $\mathcal{A}$ and $\mathcal{B}$ are orthogonal linear endomorphisms of $V$, that is, $\langle\mathcal{A}(u), \mathcal{B}(v)\rangle=0$, for all $u, v \in V$, and that $(\mathcal{A}+\mathcal{B})(v)=v$, for any $v \in \Im \mathcal{A} \oplus \Im \mathcal{B}$. Consider the endomorphism $\mathcal{A}-\alpha \mathcal{B}$, where $\alpha \in[0,1]$. Then, we have \cite{Ibort2001a}
 \begin{equation}
   \langle(\mathcal{A}-\alpha \mathcal{B})(v),(\mathcal{A}-\alpha \mathcal{B})(v)\rangle-\langle v, v\rangle=-\frac{1-\alpha}{1+\alpha}\langle\mathcal{A}(v)-\alpha \mathcal{B}(v)-v, \mathcal{A}(v)-\alpha \mathcal{B}(v)-v\rangle,
 \end{equation}
 for any $v \in \Im \mathcal{A} \oplus \Im \mathcal{B}$.

{
\begin{theorem}[Standard Carnot's theorem]
Let $T(v)$ denote the kinetic energy for the velocity $v$, namely $T(v) = \frac{1}{2} g(v, v)$,
then 
$$T_+(t) - T_-(t) = - \frac{1-\alpha}{1+\alpha}T_l(t),$$
 where $T_-(t)=T(\dot \sigma(t^-)),\ T_+(t)=T(\dot \sigma(t^+))$ and $T_l(t)=T(\dot \sigma(t^+)-\dot \sigma(t^-))$.
\end{theorem}
See \cite{Ibort2001a} for the proof.
}

Suppose that $\Psi(\sigma(t_1))=0$ and $\Psi(\sigma(t))>0$ for $t\neq t_1$. Let $\sigma(t)= (q^i(t), \dot q^i(t), z(t))$.
Since the potential has no discontinuities, the instantaneous change of the energy is given by
\begin{equation}
  E_L\left(q(t^+), \dot q(t^+), z(t^+)\right) - E_L\left(q(t^-), \dot q(t^-), z(t^-)\right)
  = - \frac{1-\alpha}{1+\alpha}T_l, \label{jump_energy}
\end{equation}
so
\begin{equation}
  E_L\left(q(t^+), \dot q(t^+), t^+\right) 
  = V\left(q(t), \dot q(t), t\right) - \frac{1-\alpha}{1+\alpha}T_l.
\end{equation}
 Let the initial conditions be $q(t_0)=q_0$, $\dot q(t_0)=\dot q_0$ and $z(t_0)=z_0$.
 Let $q(t_1^-)=q_1$, $\dot q(t_1^-)=\dot q_1$ and $z(t_1^-)=z_1$.
  Then, from Eqs.~\eqref{Eq_Lagrangian_flow_energy} and \eqref{jump_energy}, we obtain the following result.

\begin{proposition}
Suppose that $\Psi(\sigma(t_1))=0$ and $\Psi(\sigma(t))>0$ for $t\neq t_1$. Then, the evolution of the energy is given by
\begin{equation}
  E_L (q(t), \dot q(t), z(t)) = \begin{cases} e^{\left(\int_{t_0}^t -\Reeb_L(E_L) (q(\tau), \dot q(\tau), z(\tau)) \dd \tau \right)}  E_L (q_0, \dot q_0, z_0),\qquad\qquad\qquad
   \text{if } t_0<t<t_1^-,\\
  e^{\left(\int_{t_1}^t -\Reeb_L(E_L) (q(\tau), \dot q(\tau), z(\tau)) \dd \tau \right)}\left[ V (q_1, \dot q_1, z_1)-\dfrac{1-\alpha}{1+\alpha}T_l(t_1) \right],\quad
   \text{if } t>t_1^+.
\end{cases}\nonumber
\end{equation}
\end{proposition}
\subsection{Application: Rolling cylinder on a spring {plane} with an external force}
\label{cylinder_Carnot}




Consider a cylinder constrained to be above on a {plane}, in a gravitational field. Assume that the system is externally influenced by a force that depends linearly on the velocity between the cylinder and the {plane (see \cite{Hagerty2001})}. Subsequently, the cylinder may spin and slide when in contact with the {plane}. For simplicity we assume the mass distribution of the cylinder is uniform along its height. For the motion of the {plane} we include a restoring force. We model the flexible {plane} as a large mass $M$, attached to a spring with spring constant $k$. We define the stance phase of the cylinder as the dynamics of the cylinder in contact with the {plane} and the aerial phase as the dynamics of the cylinder not in contact with the {plane}.

The configuration space for the free motion of the cylinder is $Q=\mathbb{R}^2\times\mathbb{S}^1\times\mathbb{R}$. Let $m$, $I$ and $r$ be the mass, the rotational inertia about the center of mass, and the radius of the cylinder, respectively. Denote by $\gamma$ the distance from the center of mass to the center of the cylinder. Let $(x,y)\in\mathbb{R}^2$ denote the horizontal and vertical position of the cylinder's center of mass, $\phi$ the angle through which the cylinder rotates about the center of mass and $h$ the vertical displacement of the rough {plane} from its equilibrium. For a positive constant $\beta$, the contact Lagrangian for the system $L:TQ\times\mathbb{R}\to\mathbb{R}$ is given by 

$$L(x,y,\phi,h,\dot{x},\dot{y},\dot{\phi},\dot{h},z)=\frac{1}{2}m(\dot{x}^2+\dot{y}^2)+\frac{1}{2}M\dot{h}^2+\frac{1}{2}I\dot{\phi}^2-\frac{1}{2}kh^2-mgy-Mgh+\beta z.$$ 

During the stance phase (rolling motion) the system is constrained as 
\begin{equation}\label{constraints}
    x=\gamma\sin\phi+r\phi+x_0,\,\, y=h+r+\gamma\cos\phi,
\end{equation}
together with the holonomic one-side physical constraint that the cylinder cannot pass through the springly {plane}, that is, $\Psi(q)=y-h-\gamma\cos\phi+r\sin\psi\geq 0$, $\forall\psi\in \mathbb S^1$ where $q=(x,y,\phi,h)\in Q$. In particular, no force is required to impose the constraint unless we have the equality. The equality is maintained for $\psi=\pi$, while the cylinder is constrained to roll. 

The rolling constraints \eqref{constraints} can be written as one forms whose vanishing realizes the constraint: $\omega^1=\dd x-(r+\gamma\cos\phi)\dd \phi$, $\omega^2=\dd y-\dd h+(\gamma\sin\phi)\dd \phi$. So, $\omega^1$ and $\omega^2$ define the distribution $\overline{\mathcal{D}}$ describing the rolling constraints as $\overline{\mathcal{D}}=\{v\in T_{q}Q|\omega^1(v)=0,\,\omega^2(v)=0\}$.
Since $\omega^1$ and $\omega^2$ come from differentiating the holonomic constraints \eqref{constraints}, $\dd \omega^1=0=\dd \omega^2$, they are integrable. Let us define the submanifold $N\subset Q$ given by the vanishing of the integrable constraints $\omega^1$ and $\omega^2$. In addition, the {plane} imposes positive forces on the cylinder to help to enforce the constraint $\Psi(q)=0$. The one-sided nature of the constraint is realized by two positive normal forces between the cylinder and the platform, $\mu_1(q,v)$ and $\mu_2(q,v)$, which will be specified later. Then we define the constraint distribution $\mathcal{D}$, which includes the restriction of the normal forces as $$\mathcal{D}=\{v\in T_{q}Q|\omega^1(v)=0,\,\omega^2(v)=0,\, \mu_1(q,v)>0,\,\mu_2(q,v)>0\}.$$

In this context, the stance phase occurs when the cylinder is in contact with the {plane} and there are positive normal forces between the {plane} and the cylinder, that is, $\Psi(q)=0$ and $v\in\mathcal{D}_q$. Similarly, the aerial phase occurs when the cylinder is above the platform, that is, $\Psi(q)>0$. During this phase the equations of motion for the cylinder are just $m\ddot{x}-\beta \dot x=0$, $m\ddot{y}- \beta \dot y=-mg$, $I\ddot{\phi}-\beta \dot \phi=0$, where no forces are required to satisfy the one-side holonomic constraint, and then the cylinder shows circular motion about the center of mass while falling in a gravitational field.

For the stance phase motion, we compute the equations of motion by using constrained Helgotz equations \eqref{Herglotz_eq_constrained}: \begin{align*}
m\ddot{x}-m\beta\dot{x}=&\mu_1,\\
m\ddot{y}+mg-m\beta\dot{y}=&\mu_2,\\
I\ddot{\phi}-I\beta\dot{\phi}=&-(r+\gamma\cos\phi)\mu_1-\mu_2\gamma\sin\phi,\\
m\ddot{h}+Mg+kh-M\beta\dot{h}=&-\mu_2,
\end{align*}together with the constraints $x=\gamma\sin\phi+r\phi+x_0$, $y=h+r+\gamma\cos\phi$, $y-\gamma\cos\phi+r\sin\psi\geq h$. By differentiating these last three constraints,  after some computations, and solving for the Lagrange multipliers (i.e., the normal forces between the {plane} and the cylinder) $\mu_1$ and $\mu_2$ as functions of $h,\phi,\dot{h},\dot{\phi}$, we obtain 

\begin{align*}
    \mu_1=&-\Upsilon(\phi)\left(2\beta\dot{\phi}(\gamma\cos\phi+r)
    \right.\\&\left.
    +\gamma\sin\phi\dot{\phi}^2\right)+\frac{\Upsilon(\phi)mM\gamma^2\sin\phi}{\Gamma(\phi)}\left(-2\beta\sin\phi\dot{\phi}+\frac{kh}{\gamma M}+\cos\phi\dot{\phi}^2\right),\\
    \mu_2=&-\frac{mMI}{\Gamma(\phi)}\left(-2\beta\gamma\sin\phi\dot{\phi}-\frac{kh}{M}-\gamma\cos\phi\dot{\phi}^2+I^{-1}(r+\gamma\cos\phi)\mu_1\right),
\end{align*} 
with 
$\displaystyle{\Upsilon(\phi)=\frac{mI\Gamma(\phi)}{m^2M\gamma\sin\phi(r+\gamma\cos\phi)+\Gamma(\phi)(I+m(r+\gamma\cos\phi)^2)}}$ and\\ $\Gamma(\phi)=I(M+m)-mM\gamma^2\sin^2\phi$.

From Eq.~\eqref{projector_Q_local}, we have that
\begin{equation}
  \tilde{\mathcal {Q}} (X) = 
  \frac{X^y-X^h+\gamma \sin \phi X^\phi}{\frac{1}{m}-\frac{1}{M}+\frac{1}{I} (\gamma \sin \phi)^2} \left(\frac{1}{m} \frac{\partial  } {\partial y} -\frac{1}{M} \frac{\partial  } {\partial h} + \frac{1}{I} \gamma \sin \phi \frac{\partial  } {\partial \phi}  \right),  
\end{equation}
where $X=X^x \frac{\partial  } {\partial x} + X^y \frac{\partial } {\partial y} + X^h \frac{\partial  } {\partial h} + X^\phi  \frac{\partial  } {\partial \phi}$,
so Eq.~\eqref{change_velocity_projector} yields
\begin{equation}
\begin{aligned}
  &\dot x^+ = \dot x^-,\\
  &\dot y^+ = \dot y^- - (1+\alpha)  \frac{1}{m} 
  \frac{\dot y^- - \dot h^- +\gamma \sin \phi_1 \dot \phi^-}{\frac{1}{m}-\frac{1}{M}+\frac{1}{I} (\gamma \sin \phi)^2},\\
  &\dot h^+ = \dot h^- + (1+\alpha)  \frac{1}{M} 
  \frac{\dot y^- - \dot h^- +\gamma \sin \phi_1 \dot \phi^-}{\frac{1}{m}-\frac{1}{M}+\frac{1}{I} (\gamma \sin \phi)^2},\\
  &\dot \phi^+ = \dot \phi^- - (1+\alpha)  \frac{\gamma \sin \phi}{I} 
  \frac{\dot y^- - \dot h^- +\gamma \sin \phi_1 \dot \phi^-}{\frac{1}{m}-\frac{1}{M}+\frac{1}{I} (\gamma \sin \phi)^2}.
\end{aligned}
\end{equation}
Here $\dot x^\pm = \dot x(t_i^\pm),\ \dot y^\pm = \dot y(t_i^\pm),\  \dot h^\pm = \dot h(t_i^\pm),\ \dot \phi^\pm = \dot 
\phi(t_i^\pm)$ and $\phi=\phi(t_i)$, where $t_i$ is the time instant at which the impact occurs. Let $E_0=E(q(0), \dot q(0), z(0))$ be the initial value of the energy. We have that
\begin{equation}
  E_L(q(t), \dot q(t), z(t)) = e^{\beta t} E_0, 
\end{equation}
for $t<t_i$, and
\begin{equation}
  E_L(q(t), \dot q(t), z(t)) = e^{\beta t} 
  \left[V(q_1, \dot q_1, z_1) - \frac{1-\alpha}{1+\alpha} T_l(t_1)  \right],
\end{equation}
for $t>t_i$, where
\begin{equation}
  V(q_1, \dot q_1, z_1) = \frac{1}{2}k h_1^2+mgy_1+Mgh_1-\beta z_1,
\end{equation}
and
\begin{equation}
\begin{aligned}
  T_l(t_1)&=\frac{1}{2}m \left[ \left(\dot x^+ - \dot x^-  \right)^2 + \left(\dot y^+ - \dot y^-  \right)^2  \right]+\frac{1}{2}M \left(\dot h^+ - \dot h^-  \right)^2+\frac{1}{2}I \left(\dot{\phi}^+ - \dot{\phi}^-  \right)^2\\
  &=\frac{1}{2} (1+\alpha)^2  
  \left[  \frac{\dot y^- - \dot h^- +\gamma \sin \phi_1 \dot \phi^-}{\frac{1}{m}-\frac{1}{M}+\frac{1}{I} (\gamma \sin \phi)^2}\right]^2
  \left(\frac{1}{m}+\frac{1}{M}+\frac{1}{I} (\gamma \sin \phi)^2  \right).
\end{aligned}
\end{equation}

\section{Contact Lagrangian systems with instantaneous nonholonomic constraints}
\label{section_instantaneous}
The concept of distribution can be generalized by not requiring its rank to be constant (see Refs.~\cite{Cortes2001b,Vaisman1994a}). This is useful to characterize geometrically systems subject to constraints which are ``degenerate'' at certain points.  


More specifically, by a \emph{generalized distribution} on $Q$ we mean a family of vector subspaces $\mathcal{D}=\left\{\mathcal{D}_{q}\right\}$ of the tangent spaces $T_{q} Q$. Such a distribution is called \emph{differentiable} ($\mathcal{C}^\infty$) if for every $q \in \mathrm{dom}\mathcal{D}$, there is a finite number of differentiable vector fields $X_1,\ldots, X_s\in \mathcal{D}$ such that $\mathcal{D}_q = \spn \left\{X_1|_q,\ldots, X_s|_q  \right\}$.

We define the \emph{rank} of $\mathcal D$ at $q\in Q$ as the function $\rho(q)=\dim \mathcal D_q$. Observe that, if $\mathcal{D}$ is differentiable $\rho(q)$ cannot decrease in a neighbourhood of $q$, and hence $\rho$ is a lower semi-continuous function. Clearly, $\mathcal{D}$ is a distribution in the usual sense if and only if $\rho(q)$ is a constant function. In the general case, $q\in Q$ will be called a \emph{regular point} if $\rho(q)$ takes a constant value on an open neighbourhood of $q$ (in other words, $q$ is a local maximum of $\rho(q)$), and a \emph{singular point} otherwise. Obviously, the set $\mathscr{R}$ of the regular points of $\mathcal{D}$ is open. Moreover, it is dense.
Indeed, if $q_0\in Q\backslash \mathscr{R}$ and $U$ is a neighbourhood of $q_0$, $\rho|_U$ must have a maximum (since it is integer-valued and bounded), and hence $U$ contains regular points, i.e., $q_0\in \overline{\mathscr{R}}$. However, $\mathscr R$ is not connected in general.

Similarly, a \emph{generalized codistribution} on $Q$ is a family of vector subspaces $\mathcal{S}=\left\{\mathcal{S}_{q}\right\}$ of the cotangent spaces $T_{q}^{*} Q$. Such a codistribution is called \emph{differentiable} if for every $q \in \mathrm{dom}\mathcal{S}$, there is a finite number of differentiable 1-forms $\omega_{1}, \ldots, \omega_{s}\in \mathcal S$ such that $\mathcal S_{q}=\spn\left\{\omega_{1}\left(q\right), \ldots, \omega_{s}\left(q\right)\right\}$. The regular and singular points are defined analogously to the ones in generalized distributions. Similarly, the set of regular points of $\mathcal S$ is open, dense and, generally, non-connected.

\begin{example}
Let $Q=\RR^2$ and $\mathcal{D}_{(x,y)}=\spn \left\{\partial/\partial x, \varphi(y) \partial/\partial y  \right\}$, where $\varphi(y)$ is a $\mathcal{C}^\infty$-function with $\varphi(y)=0$ for $y\leq 0$ and $\varphi(y)>0$ for $y>0$. For instance,
\begin{equation}
  \varphi(y)=
  \left\{
    \begin{array}{ll}
    0, &y\leq 0,\\
     e^{-1/y^2}, & y> 0.
    \end{array}
    \right.
\end{equation}
Then, the singular points are those of the $x$-axis, while the connected components are the half-{plane}s $y<0$ (where $\rho = 1$) and $y>0$ (where $\rho =2$).
\end{example}

Given a generalized distribution, $\mathcal D$, we define its annihilator, $\mathcal D^{\circ}$, as the generalized codistribution given by
\begin{equation}
\begin{aligned}
  \mathcal D^{\circ}: \operatorname{dom} \mathcal D \subset Q & \rightarrow T^* Q \\
  q & \mapsto \mathcal D_{q}^{\circ}=\left(\mathcal D _{q}\right)^{\circ}.
\end{aligned}
\end{equation}



Let $\mathcal D^\circ$ be a generalized differentiable codistribution on $Q$.
 The codistribution induces a decomposition of $Q$ into regular and singular points, namely, $Q= R\cup S$. Let us fix $R_c$, a connected component of $R$. Consider the restriction of the codistribution to $R_c$, $\mathcal D_c^\circ=\mathcal D_{\mid R_c}^\circ: R_c\subset Q\to T^*Q$. Clearly, $\mathcal D_{\mid R_c}^\circ$ is a regular codistribution. Let $\mathcal D_c:R_c\to TQ$ be the annihilator of $\mathcal D_c^\circ$. Now, we can consider a contact Lagrangian $L$ on $TQ\times \RR$ constrained to $\mathcal D_c$ and apply the theory for non-holonomic contact Lagrangian systems developed by de León, Jiménez and Lainz in \cite{deLeon2021h}. In this way, the problem can be solved in each connected component of $R$. However, if the motion reaches a singular point, the rank of $\mathcal D$ can vary suddenly, and the equations of motion can no longer be derived from the Herglotz principle with constraints. As a matter of fact, an impulsive force can emerge because of the change of rank of $\mathcal D$. 
{ The rank of $\mathcal D^\circ$, i.e., the corank of $\mathcal D$, is equal to the number of constraints.}
 Consider a trajectory of the system $q(t)$ which reaches a singular point at $t_0$, i.e., $q(t_0)\in S$, such that $q((t_0- \varepsilon, t_0))\subset R$ and $q((t_0, t_0+\varepsilon))\subset R$ for sufficiently small $\varepsilon>0$. Let 
 $\rho_\pm=\rho(q(t_0\pm\varepsilon))$ and $\rho_0=\rho(q(t_0))$.
 Recall that regular points are local maximums of $\rho$, and hence there are three possible cases that can occur at $t_0$:
 \begin{enumerate}
 \item $\rho_-=\rho_0<\rho_+$,
 \item $\rho_->\rho_0=\rho_+$,
 \item $\rho_->\rho_0$ and $\rho_+>\rho_0$.
 \end{enumerate}
In the first and third cases, the trajectory must satisfy, immediately after the point $q(t_0)$, {a greater number of constraints} which were not present before. This leads to an impulsive force which imposes the new constraints on the motion. Hereinafter, assume that we are in one of this two cases, i.e., that $\rho_0<\rho_+$.

Given the codistribution $\mathcal D^\circ $, we can introduce an associated distribution $\mathcal{D}^\ell$ on $TQ\times \RR$, whose annihilator is given by 
\begin{equation}
  \mathcal{D}^{\ell^\circ } = \left(\tau_Q\circ \prtr  \right)^* \mathcal{D}^\circ,
\end{equation}
where $\tau_Q:TQ\to Q$ is the canonical projection and $\prtr:TQ\times \RR\to TQ$ denotes the projection on the first component.

\begin{theorem}
Let $\sode_{L,\mathcal{D}}$ be a vector field on $TQ\times \RR$ such that
\begin{subequations}
\begin{flalign}
  &\flat_L (\sode_{L,\mathcal{D}})- \dd E_L + \left(E_L + \Reeb_L(E_L)  \right) \eta_L \in \mathcal{D}^{\ell^\circ},\\
 &\Im \sode_{L,\mathcal{D}}|_{\mathcal D\times \RR} \subset {T} \left( \mathcal D\times \RR \right).
\end{flalign}
\label{eqs_sode_constrained}
\end{subequations}
Then:
\begin{enumerate}
\item $\sode_{L,\mathcal{D}}$ is a SODE,
\item the integral curves of $\sode_{L,\mathcal{D}}$ are solutions of the contrained Herglotz equations \eqref{Herglotz_eq_constrained}.
\end{enumerate}
\end{theorem}

\begin{proof}
  Since $\mathcal D^\circ$ is a differentiable codistribution, for each $q\in Q$ there exists a local neighbourhood such that
  $\mathcal D_{\mid U}=\spn \left\{\psi^1, \ldots, \psi^m  \right\}$ for some 1-forms $\psi^i$ on $Q$, where $m$ is the local maximum of $\rho$ at $U$.
  The rest of the proof is identical to the one of Theorem 6 from Ref.~\cite{deLeon2021h}.
\end{proof}

Let $\mathcal S$ be the generalized distribution on $TQ\times \RR$ defined by $\sharp_L (\mathcal D^{\ell^\circ})$, where $\sharp_L = \flat_L^{-1}$. Let $Y_a$ be the local vector fields on $TQ\times \RR$ given by
\begin{equation}
  \flat_L (Y_a) = \tilde \psi^a,
\end{equation}
 where $\tilde \psi^a = \psi^a_i \dd q^i$ are 1-forms on $TQ\times \RR$.
Clearly, $\mathcal S_q$ is generated by $\left\{Y_a|_q  \right\}$.
We have that \cite{deLeon2021h}
\begin{equation}
  Y_a = -W^{ij} \psi^a_j \frac{\partial  } {\partial \dot q^i}.
\end{equation}
%
Consider the condition
\begin{equation}
  \mathcal{S} \cap T \left(\mathcal D\times \RR  \right) = \left\{0  \right\}. \label{condition_uniqueness}
\end{equation}

Let $X=X^b Z_b$ be a vector field tangent to $\mathcal S$. Then $X\in T\left(\mathcal D\times \RR\right)$ if and only if
\begin{equation}
  0 = \dd \bar \psi^a(X) = \psi^a_i (X^b Z_b)^i = -\psi^a_i W^{jk} \psi^b_k X^b,
\end{equation}
where $\bar \psi^b=\psi^b_i(q)\dot q^i\ (b=1,\ldots, r)$  are functions on $Q\times \RR$.
Consider the matrix
\begin{equation}
  \left(\mathcal{C}_{ab}\right) = -\left( W^{ij} \psi^a_i \psi^b_j  \right). 
\end{equation}

For each $(v_q, z)\in \mathcal D\times \RR$, we have that $\dim \mathcal S_{(v_q, z)}=\rho(q)$ and $\dim T_{(v_q,z)}(\mathcal D \times \RR)=2n+1-\rho (q)$. 
Therefore, if the condition \eqref{condition_uniqueness} holds, we have that
\begin{equation}
  \mathcal S \oplus T (\mathcal D \times \RR) = T_{\mathcal D \times \RR} (TQ \times \RR),
\end{equation}
where $T_{\mathcal D \times \RR} (TQ \times \RR)$ consists of the tangent vectors of $TQ\times \RR$ at points of $\mathcal D \times \RR$.
 Then it is natural to introduce the following projectors:
\begin{subequations}
\begin{flalign}
&\hat{\mathcal{P}}: T_{\mathcal{D} \times \RR}(T Q \times \RR) \rightarrow T(\mathcal{D} \times \RR), \\
&\hat{\mathcal{Q}}: T_{\mathcal{D} \times \RR}(T Q \times \RR) \rightarrow \mathcal{S}.
\end{flalign}
\end{subequations}
Let $X= \hat{\mathcal{P}}(\sode_L|_{\mathcal D \times \RR})$. By construction, $\Im X\subset T(\mathcal D\times \RR)$. On the other hand, at the points in $\mathcal D \times \RR$, we have that
\begin{equation}
\begin{aligned}
  b_{L}(X)-\mathrm{d} E_{L}+\left(E_{L}+\mathcal{R}_{L}\left(E_{L}\right)\right) \eta_{L} 
  &=b_{L}\left(\sode_{L}-\hat{\mathcal{Q}}\left(\sode_{L}\right)\right)-\mathrm{d} E_{L}+\left(E_{L}+\mathcal{R}_{L}\left(E_{L}\right)\right) \eta_{L} \\
  &=-b_{L}\left(\hat{\mathcal{Q}}\left(\sode_{L}\right)\right) \in \mathcal{D}^{\ell^{\circ}},
\end{aligned}
\end{equation} 
so $X= \sode_{L,\mathcal D}$.

We shall now compute an explicit expression of $\sode_{L, \mathcal{D}}$. Let $Y$ be a vector field on $TQ \times \RR$. Then, choosing a local basis $\left\{\beta_{i}\right\}$ of $T(\mathcal{D} \times\RR)$ we may write the restriction of $Y$ to $\mathcal{D} \times \RR$ as 
\begin{equation}
  Y|_{\mathcal{D} \times \RR}=Y^{i} \beta_{i}+\lambda^{a} Z_{a},
\end{equation}
so
\begin{equation}
  \dd \bar \psi^b (Y) = \lambda^a \mathcal{C}_{ba},
\end{equation}
and hence the coefficients $\lambda^a$ are given by
\begin{equation}
  \lambda^a = \mathcal{C}^{ba} \dd \bar \psi^b(Y),
\end{equation}
where $(\mathcal C^{ab}) = (\mathcal C_{ab})^{-1}$. Thus
\begin{equation}
\begin{aligned}
  &\hat{\mathcal{Q}}\left(Y_{\mid \mathcal{D} \times \RR}\right)=\mathcal C^{b a} \mathrm{~d} \bar{\psi}^{b}(Y) Z_{a},\\  
  &\hat{\mathcal{P}}\left(Y_{\mid \mathcal{D} \times \RR}\right)=Y_{\mid \mathcal{D} \times \RR}-\mathcal C^{b a} \mathrm{~d} \bar{\psi}^{b}(Y) Z_{a}.
\end{aligned}
\end{equation}
We then obtain the following result.
\begin{proposition}
If $\sode_L$ is the Hamiltonian vector field associated  with the energy $E_L$, then 
\begin{equation}
  \sode_{L, \mathcal{D}}=\sode_{L}|_{\mathcal{D} \times \RR}-\mathcal C^{b a} \mathrm{~d} \bar{\psi}^{b}\left(\sode_{L}\right) Z_{a}.
\end{equation}
\end{proposition}


Let us consider now a contact Lagrangian system with mechanical Lagrangian $L(q, \dot q, z)=\frac{1}{2}g(\dot q, \dot q)-V(q,z)$. 
Let us introduce the following vector subspaces of $T^*_{q(t_0)}Q$:
\begin{equation}
\begin{aligned}
  & \mathcal D_{q(t_0)}^{\circ -} 
  \coloneqq \left\{\alpha \in T^*_{q(t_0)}Q \mid 
  \exists \tilde \alpha: (t_0-\varepsilon, t_0)\to T^*Q \text{ such that }
  \tilde \alpha(t)\in \mathcal D_{q(t)} 
  \text{ and } \lim_{t\to t_0^-} \tilde \alpha(t) = \alpha
     \right\},\\
  & \mathcal D_{q(t_0)}^{\circ +}
  \coloneqq \left\{\alpha \in T^*_{q(t_0)}Q \mid 
  \exists \tilde \alpha: (t_0, t_0+\varepsilon)\to T^*Q \text{ such that }
  \tilde \alpha(t)\in \mathcal D_{q(t)} 
  \text{ and } \lim_{t\to t_0^+} \tilde \alpha(t) = \alpha
     \right\}.  
\end{aligned}
\end{equation}
We have that
\begin{equation}
  \left(\mathcal D_{q(t_0)}^{\circ \pm}  \right)^\perp
  = \lim_{t\to t_0^\pm} \left(\mathcal D_{q(t)^\circ}  \right)^\perp,
\end{equation}
where the superscript $\perp$ denotes the orthogonal complement with respect to the bilinear form induced by the metric $g$ on $T_{q(t_0)}^*Q$, and the limits $(D^\perp)^\pm$ are defined as in the case of $D^\pm$. In what follows, we shall also denote by $g$ the bilinear form induced by the metric on $T_{q(t_0)}^*Q$. In each case, the meaning should be clear by the context.

From Eq.~\eqref{constrained_Herglotz_form}, we have that
\begin{equation}
   \left(\frac{\mathrm{d}p_i} {\mathrm{d}t} - \frac{\partial L} {\partial q^i} - \frac{\partial L} {\partial \dot q^i} \frac{\partial L} {\partial z}  \right)\ 
   \dd q^i
   \in \mathcal D_{q(t)}^\circ ,
\end{equation}
so
\begin{equation}
 \lim_{t\to t_0^+} \int_{t_0}^t  \left(\frac{\mathrm{d}p_i} {\mathrm{d}t} - \frac{\partial L} {\partial q^i} - \frac{\partial L} {\partial \dot q^i} \frac{\partial L} {\partial z}  \right)\ 
   \dd q^i\ \dd t
   = \left[p_i(t_0^+)-p_i(t_{0})  \right] \dd q^i
   \in \lim_{t\to t_0^+} \mathcal D_{q(t)}^\circ
   = \mathcal D_{q(t)}^{\circ +}.
\end{equation}
Additionally, the `post-impact' momentum $p(t_0^+)$ must satisfy the new constraints imposed by $\mathcal D_{q(t_0)}^{\circ +}$. Therefore, the change of momentum is determined by the following equations:
\begin{subequations}
\begin{flalign}
    & \left[p_i(t_0^+)-p_i(t_{0})  \right] \dd q^i
   \in \mathcal D_{q(t)}^{\circ +},
   \label{change_momentum_distribution}
   \\ 
   & p_i(t_0^+)\ \dd q^i \in \left(\mathcal D_{q(t)}^{\circ +}  \right)^\perp.
\end{flalign}
\end{subequations}
A momentum jump occurs if the `pre-impact' momentum does not satisfy the constraints imposed by $\mathcal D_{q(t)}^{\circ +}$, that is,
\begin{equation}
   p_i(t_0^-)\ \dd q^i \notin \left(\mathcal D_{q(t)}^{\circ +}  \right)^\perp.
\end{equation}

Let $m=\max\{\rho_-, \rho_+\}$. Then, there exists a neighbourhood $U$ of $q(t_0)$ and 1-forms $\psi^1, \ldots, \psi^m$ such that $\mathcal D_q = \spn \left\{\psi^i(q)  \right\}_{i=1}^m$ for any $q\in Q$.
 Let us suppose that $\psi^1, \ldots, \psi^{\rho_+}$ are linearly independent at the regular posterior points. Obviously, these 1-forms are linearly dependent at $q(t_0)$. In order to simplify the notation, hereinafter let $\psi^a=\psi^a(q(t))$. We have that
 \begin{equation}
   \psi^a_i\dot q^i(t)
   =\psi^a_{i} g^{ij} p_j(t) = 0
   \label{constraint_local}
 \end{equation}
 for $a=1, \ldots, \rho_+$. The metric $g$ induces the decomposition $T_q^*Q= \mathcal D_q^\circ \oplus \mathcal D_q^{\circ \perp}$, with the projectors
 \begin{equation}
 \begin{aligned}
   &\mathcal P_q : T_q^*Q \to \mathcal D_q^{\circ \perp},\\
   & \mathcal Q_q : T_q^*Q \to \mathcal D_q^{\circ}.
 \end{aligned}
 \end{equation}
 Consider the matrix $(\mathcal C^{ab}) = (\psi^a_i g^{ij} \psi^b_j)$, in other words, $\mathcal C= \psi g^{-1} \psi^T$, and let $(\mathcal C_{ab})$ denote its inverse matrix. The projector $\mathcal P_q$ is given by
 \begin{equation}
   \mathcal P_q(\alpha) 
   = \alpha -\mathcal C_{ab} g^{ij} \psi^a_i \alpha_j \psi^b,
 \end{equation}
 for $\alpha=\alpha_i \dd q^i \in T_q^*Q$.
 From Eq.~\eqref{constraint_local}, we have that $\mathcal P_{q(t)}\left(p_i(t)\dd q^i_{\mid q(t)}\right)=p_i(t)\dd q^i_{\mid q(t)} $, so
 \begin{equation}
   p_i(t_0^+) \dd q^i_{\mid q(t_0^+)}
   = \lim_{t\to t_0^+} p_i(t) \dd q^i_{\mid q(t)}
   = \lim_{t\to t_0^+} \mathcal P_{q(t)}\left(p_i(t_0^+)\dd q^i_{\mid q(t)}\right)
   \in \left( \mathcal D_{q(t_0)}^{\circ +}  \right)^\perp.
 \end{equation}
Now, Eq.~\eqref{change_momentum_distribution} implies that
\begin{equation}
   \lim_{t\to t_0^+} \mathcal P_{q(t)}\left( \left(p_i(t_0+) - p_i(t_0)  \right)\dd q^i_{\mid q(t)}\right) = 0,
\end{equation}
and hence the change of momentum is given by
\begin{equation}
  p_i(t_0^+) \dd q^i_{\mid q(t_0)}
  =  \lim_{t\to t_0^+} \mathcal P_{q(t)} \left(  p_i(t_0^-) \dd q^i_{\mid q(t_0)} \right).
  \label{eq_change_momentum_P_matrix}
\end{equation}
Locally, this can be expressed as
\begin{equation}
  p_i(t_0^+) = p_i(t_0^-) - \lim_{t\to t_0^+} \sum_{a,b,j,k}  
  \left.\left(\mathcal C_{ab} \psi^a_j g^{jk} \psi^b_i\right)\right|_{q(t)} p_k(t_0^-),
\end{equation}
or, in matrix form,
\begin{equation}
  p(t_0^+)= \left[\operatorname{Id}- \lim_{t\to t_0^+} \left. \left(\psi^T \mathcal C^{-1} \psi g^{-1}  \right)\right|_{q(t)}  \right] p(t_0^-).
  \label{eq_change_momentum_C_matrix}
\end{equation}

\begin{example}[Rolling cylinder on a spring {plane} with an external force]
Consider the system from Subsection \ref{cylinder_Carnot}. From the constraints \eqref{constraints}, we obtain the codistribution
\begin{equation}
  \mathcal D_{(x, y, \phi, h)}^\circ=\left\{\begin{array}{ll}
  \{0\}, & \text { if } y-h-\gamma \cos \phi>0, \\
  \operatorname{span}\left\{\dd x-(r+\gamma\cos\phi)\dd \phi, \dd y-\dd h+(\gamma\sin\phi)\dd \phi\right\}, & \text { if } y-h-\gamma \cos \phi=0.
\end{array}\right.
\end{equation}

The instantaneous change of momentum that occurs when the cylinder transitions from the aerial phase to the stance phase can be computed via Eq.~\eqref{eq_change_momentum_P_matrix}. The constraints can be expressed in matrix form as
\begin{equation}
    \phi = \left(
     \begin{array}{cccc}
      1 & 0 & -\gamma  \cos (\phi )-r & 0 \\
      0 & 1 & -\gamma  \sin (\phi ) & -1 \\
     \end{array}
    \right)
\end{equation}
so the matrix $\mathcal C$ is given by
\begin{align*}
    \mathcal C &=\phi g^{-1} \phi^T 
    = \left(
    \begin{array}{cc}
     \frac{(-\gamma  \cos (\phi )-r)^2}{I}+\frac{1}{m} & -\frac{\gamma  \sin (\phi ) (-\gamma  \cos (\phi )-r)}{I} \\
     -\frac{\gamma  \sin (\phi ) (-\gamma  \cos (\phi )-r)}{I} & \frac{\gamma ^2 \sin ^2(\phi )}{I}+\frac{1}{m}+\frac{1}{M} \\
    \end{array}
    \right)
\end{align*}
and the projector $\mathcal P$ is
{
\begin{align*}
   & \mathcal P
    = \operatorname{Id}_4-\phi^T \mathcal C^{-1} \phi g^{-1} \\ 
\end{align*}}
Hence, 
\begin{align*}
    p_x^+&=\frac{2 m (\gamma  \cos (\phi )+r) [\gamma  \sin (\phi ) (m p_h^--M p_y^-)+(m+M) (\gamma  p_x^- \cos (\phi )+p_x^- r+p_\phi^-)]}{2 (m+M) \left(I+m r^2\right)+\gamma ^2 m (m+2 M)+\gamma  m (\gamma  m \cos (2 \phi )+4 r (m+M) \cos (\phi ))},\\
    p_y^+&=\frac{m}{(m+M) \left(I+m r^2+\gamma  m \cos (\phi ) (\gamma  \cos (\phi )+2 r)\right)+\gamma ^2 m M \sin ^2(\phi )}\\
    &\times\left[(p_h^-+p_y^-) \left(I+m r^2+\gamma  m \cos (\phi ) (\gamma  \cos (\phi )+2 r)\right)
    \right.\\&\left.
    -\gamma  M \sin (\phi ) (\gamma  p_x^- \cos (\phi )+p_x^- r+p_\phi^-)+\gamma ^2 M p_y^- \sin ^2(\phi )\right],\\
    p_\phi^+&=\frac{2 I (\gamma  \sin (\phi ) (m p_h^--M p_y^-)+(m+M) (\gamma  p_x^- \cos (\phi )+p_x^- r+p_\phi^-))}{2 (m+M) \left(I+m r^2\right)+\gamma ^2 m (m+2 M)+\gamma  m (\gamma  m \cos (2 \phi )+4 r (m+M) \cos (\phi ))},\\
    p_h^+&=\frac{M}{2 (m+M) \left(I+m r^2\right)+\gamma ^2 m (m+2 M)+\gamma  m (\gamma  m \cos (2 \phi )+4 r (m+M) \cos (\phi ))}\\
    &\times \left[2 (p_h^-+p_y^-) \left(I+m r^2\right)+\gamma  m (4 r (p_h^-+p_y^-) \cos (\phi )
     \right.\\&\left.
     +2 \sin (\phi ) (\gamma  p_x^- \cos (\phi )+p_x^- r+p_\phi^-)+\gamma  p_y^- \cos (2 \phi ))+\gamma ^2 m (2 p_h^-+p_y^-)\right].
\end{align*}
Making use of the relations
\begin{equation}
    p_x^\pm =  m \dot x^\pm,\quad
    p_y^\pm =  m \dot y^\pm,\quad
    p_\phi^\pm =  I\dot \phi^\pm,\quad
    p_h^\pm =  M \dot h^\pm
\end{equation}
one can obtain the instantaneous change of velocity.
\end{example}

\begin{example}[The rolling sphere with dissipation]
Consider a homogeneous sphere rolling on a {plane}. The configuration space is $Q=\mathbb{R}^{2} \times S O(3)$ (see \cite{Cortes2001b} for the symplectic counterpart of this example). Let $(x, y)$ denote the position of the centre of the sphere and let $(\varphi, \theta, \psi)$ denote the Eulerian angles.

Assume that the {plane} is smooth if $x<0$ and absolutely rough if $x>0$. On the smooth part, the motion of the ball is free, whereas when it reaches the rough half-{plane}, the sphere rolls without slipping. Let us suppose that the motion of the sphere, both on the smooth and rough half-{plane}s, has a dissipation linear in the velocities. The contact Lagrangian of the system is
\begin{equation}
  L = \frac{1}{2}\left[\dot{x}^{2}+\dot{y}^{2}+k^{2}\left(\omega_{x}^{2}+\omega_{y}^{2}+\omega_{z}^{2}\right)\right] - \beta Z,
\end{equation}
where $\omega_{x}, \omega_{y}$ and $\omega_{z}$ are the angular velocities with respect to the inertial frame, given by
\begin{equation}
\begin{aligned}
  &\omega_{x}=\dot{\theta} \cos \psi+\dot{\varphi} \sin \theta \sin \psi, \\
  &\omega_{y}=\dot{\theta} \sin \psi-\dot{\varphi} \sin \theta \cos \psi, \\
  &\omega_{z}=\dot{\varphi} \cos \theta+\dot{\psi},
\end{aligned}
\end{equation}
$Z$ is the `contact variable' and $\beta$ is a positive constant.

The condition of rolling without sliding is given by
\begin{equation}
\begin{array}{l}
  \phi^{1}=\dot{x}-r \omega_{y}=0, \\
  \phi^{2}=\dot{y}+r \omega_{x}=0.
\end{array}
\end{equation}
Let us introduce the quasi-coordinates $q^1,\ q^2$ and $q^3$ such that $\dot q^1 = \omega_x,\ \dot q^2= \omega_2$ and $\dot q^3= \omega_3$. The generalized distribution $\mathcal D^\circ$
characterizing the constraints has annihilator
\begin{equation}
  \mathcal D_{(x, y, \phi, \theta, \psi)}^\circ=\left\{\begin{array}{ll}
  \{0\}, & \text { if } x \leqslant 0 \\
  \operatorname{span}\left\{\mathrm{d} x-r \mathrm{~d} q^{2}, \mathrm{~d} y+r \mathrm{~d} q^{1}\right\}, & \text { if } x>0
\end{array}\right.
\end{equation}
The set of regular points of the distribution has two connected components, namely,
\begin{equation}
\begin{aligned}
   &R_1 = \left\{(x, y, \varphi, \theta, \psi )\in Q\mid x<0  \right\},\\  
   &R_2 = \left\{(x, y, \varphi, \theta, \psi )\in Q\mid x>0  \right\},
\end{aligned}
\end{equation}
while the line $\left\{x=0  \right\}$ belongs to the singular set of $\mathcal D$.
On $R_1$ the equations of motion are
\begin{equation}
\begin{aligned}
  &  \ddot x + \beta \dot x  = 0,\\
  &  \ddot y + \beta \dot y  = 0,\\
  & k^2 \dot \omega_a + \beta \omega_a = 0, \quad a=x,y,z.
\end{aligned}
\label{eqs_sphere_R1}
\end{equation}
On $R_2$ the equations of motion are
\begin{equation}
\begin{aligned}
  &  \ddot x + \beta \dot x  = \lambda_1,\\
  &  \ddot y + \beta \dot y  = \lambda_2,\\
  & k^2 \dot \omega_x + k^2 \beta \omega_x = r\lambda_2,\\
  & k^2 \dot \omega_y + k^2 \beta \omega_y = -r\lambda_1,\\
  & k^2 \dot \omega_z + k^2\beta \omega_z = 0,\\
  & \dot x - r\omega_y = 0,\\
  & \dot y+ r \omega_x = 0.
\end{aligned}
\label{eqs_sphere_R2}
\end{equation}

Assume that the sphere starts its motion at some point in $R_1$ with positive velocity in the $x$-direction, namely, $x(0)=x_0,\, y(0)=y_0,\, \omega_a(0)=(\omega_a)_0\, (a=x,y,z)$ such that $x_0<0$ and $\dot x_0>0$. Integrating Eqs.~\eqref{eqs_sphere_R1} yields
\begin{equation}
\begin{aligned}
  &  x(t)=\frac{\dot x_0}{\beta } \left(1-e^{-\beta  t}\right)+ x_0,\\
  &  y(t)=\frac{\dot y_0}{\beta }\left(1-e^{-\beta  t}\right)  +y_0,\\
  & \omega_a(t)=e^{-\frac{\beta  t}{k^2}}(\omega_a)_0, \quad a=x,y,z,
\end{aligned}
\label{eqs_sphere_R1_integrated}
\end{equation}
for $x(t)<0$. At time $\bar{t}=-x_0/\dot x_0$ the sphere reaches the rough {surface} of the plane, where the codistribution $\mathcal D^\circ$ is no longer zero, so the sphere is forced to roll without sliding.

We can compute the instantaneous change of momentum by means of Eq.~\eqref{eq_change_momentum_P_matrix}. 
The constraints can be expressed in matrix form as
\begin{equation}
    \phi = (\phi_i^a) =  \begin{pmatrix}
        1 & 0 & 0 & -r & 0\\
        0 & 1 & r & 0 & 0
    \end{pmatrix}, 
\end{equation}
so the matrix $\mathcal C$ is given by
\begin{align*}
    \mathcal C &=\phi g^{-1} \phi^T 
    = \begin{pmatrix}
        1+\frac{r^2}{k^2} & \\
      &1+\frac{r^2}{k^2},
    \end{pmatrix}
    =\frac{k^2+r^2}{k^2} \mathrm{Id}_2
\end{align*}
and the projector $\mathcal P$ is
\begin{align*}
    \mathcal P
    &= \operatorname{Id}_5-\phi^T \mathcal C^{-1} \phi g^{-1}  
    =\left(
    \begin{array}{ccccc}
     \frac{r^2}{k^2+r^2} & 0 & 0 & \frac{r}{k^2+r^2} & 0 \\
     0 & \frac{r^2}{k^2+r^2} & -\frac{r}{k^2+r^2} & 0 & 0 \\
     0 & -\frac{k^2 r}{k^2+r^2} & \frac{k^2}{k^2+r^2} & 0 & 0 \\
     \frac{k^2 r}{k^2+r^2} & 0 & 0 & \frac{k^2}{k^2+r^2} & 0 \\
     0 & 0 & 0 & 0 & 1 \\
    \end{array}
    \right).
\end{align*}
Hence, 
\begin{align*}
    &\left(p_{x}\right)_{+}=\frac{r^{2}\left(p_{x}\right)_{-}+r\left(p_{2}\right)_{-}}{r^{2}+k^{2}},\\
    &\left(p_{y}\right)_{+}=\frac{r^{2}\left(p_{y}\right)_{-}-r\left(p_{1}\right)_{-}}{r^{2}+k^{2}},\\
    &\left(p_{1}\right)_{+}=\frac{-r k^{2}\left(p_{y}\right)_{-}+k^{2}\left(p_{1}\right)_{-}}{r^{2}+k^{2}},\\
    &\left(p_{2}\right)_{+}=\frac{r k^{2}\left(p_{x}\right)_{-}+k^{2}\left(p_{2}\right)_{-}}{r^{2}+k^{2}},\\
    &\left(p_{3}\right)_{+}=\left(p_{3}\right)_{-}.
\end{align*}
We can introduce the quasi-coordinates $q^1, q^2$ and $q^3$ such that $\dot q^1= \omega_x,\, \dot q^2= \omega_y$ and $\dot q^3= \omega_z$. Then
\begin{equation}
\begin{array}{lll}
    p_{x}=\frac{\partial L}{\partial \dot x}=\dot{x}, 
    &p_{y}=\frac{\partial L}{\partial \dot y}=\dot{y}, \\ 
    p_{1}=\frac{\partial L}{\partial \dot q^1}=k^{2} \omega_{x}, 
    &p_{2}=\frac{\partial L}{\partial \dot q^2}=k^{2} \omega_{y}, 
    &p_{3}=\frac{\partial L}{\partial \dot q^3}=k^{2} \omega_{z},
\end{array}
\end{equation}
and thus the instantaneous change of velocity is given by
\begin{equation}
\begin{aligned}
    &\dot{x}_{+}=\frac{r^{2} \dot{x}_{-}+r k^{2}\left(\omega_{y}\right)_{-}}{r^{2}+k^{2}},\\
    &\dot{y}_{+}=\frac{r^{2} \dot{y}_{-}-r k^{2}\left(\omega_{x}\right)_{-}}{r^{2}+k^{2}},\\
    &\left(\omega_{x}\right)_{+}=\frac{-r \dot{y}_{-}+k^{2}\left(\omega_{x}\right)_{-}}{r^{2}+k^{2}},\\
    &\left(\omega_{y}\right)_{+}=\frac{r \dot{x}_{-}+k^{2}\left(\omega_{y}\right)_{-}}{r^{2}+k^{2}},\\
    & (\omega_z)_+=(\omega_z)_{-}.
    \label{change_velocities_sphere}
\end{aligned}
\end{equation}
We now integrate Eqs.~\eqref{eqs_sphere_R2} with the initial conditions $x(\bar{t})=x_1,\, y(\bar{t})=y_1,\, \dot x(\bar{t})=\dot x_+,\, \dot y(\bar{t})= \dot y_+,\, \dot z(\bar{t})=\dot z_+,\, \omega_{a}(\bar{t})=(\omega_a)_+\, (a=x,y,z)$.
The equations of motion in $R_2$ obtained are
\begin{align}
    x(t)&=\frac{e^{-\beta t}}{\beta ^2 \left(k^2+r^2\right)}
    \left[e^{\beta  t} \left(\lambda_1  \left(k^2+r^2\right) (\beta  t-\beta  \bar{t}-1)+\beta  k^2 (r (\omega_y)_-+\beta  x_1)
    \right.\right.\\&\left.\left.
    +\beta  r^2 (\dot x_-+\beta  x_1)\right)+e^{\beta  \bar{t}} \left(k^2 (\lambda_1 -\beta  r (\omega_y)_-)+r^2 (\lambda_1 -\beta  \dot x_-)\right)\right],\\
    y(t)&=\frac{e^{-\beta  t} }{\beta ^2 \left(k^2+r^2\right)}
    \left[e^{\beta  t} \left(\lambda_2  \left(k^2+r^2\right) (\beta  t-\beta  \bar{t}-1)+\beta  k^2 (\beta  y_1-r (\omega_x)_-)
    \right.\right.\\&\left.\left.
    +\beta  r^2 (\dot y_-+\beta  y_1)\right)+e^{\beta  \bar{t}} \left(k^2 (\lambda_2 +\beta  r (\omega_x)_-)+r^2 (\lambda_2 -\beta  \dot y_-)\right)\right],\\
    \omega_x(t)&=\frac{1}{\beta }{\lambda_2  r-\frac{e^{\frac{\beta  (\bar{t}-t)}{k^2}} \left(k^2 (\lambda_2  r-\beta  (\omega_x)_-)+\lambda_2  r^3+\beta  r \dot y_-\right)}{k^2+r^2}},\\
    \omega_y(t)&= \frac{1}{\beta }{\frac{e^{\frac{\beta  (\bar{t}-t)}{k^2}} \left(k^2 (\lambda_1  r+\beta  (\omega_y)_-)+\lambda_1  r^3+\beta  r \dot x_-\right)}{k^2+r^2}-\lambda_1  r},\\
    \omega_z(t)&=(\omega_z)_- e^{\frac{\beta  (\bar{t}-t)}{k^2}}.
\end{align}
The Lagrange multipliers $\lambda_1$ and $\lambda_2$ are obtained by imposing
$\dot x - r\omega_y = 0$ and  $\dot y+ r \omega_x = 0$, yielding
\begin{align*}
    \lambda_1&=\frac{\beta  r \left(k^2 (\omega_y)_-+r \dot x_-\right) \left(e^{\beta  \bar{t}}-e^{\frac{\beta  (\bar{t}-t)}{k^2}+\beta  t}\right)}{\left(k^2+r^2\right) \left(r^2 e^{\frac{\beta  (\bar{t}-t)}{k^2}+\beta  t}-r^2 e^{\beta  t}-e^{\beta  t}+e^{\beta  \bar{t}}\right)},\\
    \lambda_2&=\frac{\beta  r \left(r \dot y_--k^2 (\omega_x)_-\right) \left(e^{\beta  \bar{t}}-e^{\frac{\beta  (\bar{t}-t)}{k^2}+\beta  t}\right)}{\left(k^2+r^2\right) \left(r^2 e^{\frac{\beta  (\bar{t}-t)}{k^2}+\beta  t}-r^2 e^{\beta  t}-e^{\beta  t}+e^{\beta  \bar{t}}\right)}.
\end{align*}


\end{example}

\section{Conclusions and outlook}
\label{section_conclusions}
We introduced contact Lagrangian systems under impulsive
forces and constraints from a geometric point of view, as well as instantaneous nonholonomic constraints which are not uniform along the configuration space. 
{{This was done via Herglotz variational principle.}}
In addition, we have shown a  Carnot-type theorem for contact Lagrangian systems subject to impulsive forces and constraints and we have also shown that the vector field describing the dynamics of a contact Lagrangian system is determined by defining projectors to evaluate the constraints by using a Riemannian metric. We have illustrated the
theoretical results with two examples: a rolling cylinder on a springily {plane} and a rolling sphere on a non-uniform {plane}, both with dissipation. 

In a future work, we shall provide a variational approach for contact Lagrangian systems with impulsive constraints, extending the variational approach from \cite{Lopez-Gordon2022} (see also \cite{Fetecau2003}). 
Furthermore, we plan to study hybrid contact Hamiltonian and Lagrangian systems. Hybrid systems are an alternative formalism for describing physical system with impacts, as well as certain UAVs (unmanned aerial vehicles) systems and legged robots (for instance, see \cite{Colombo2022} and references therein).

\section*{Acknowledgments}
The authors acknowledge financial support from the Spanish Ministry of Science and Innovation (MCIN/AEI/ 10.13039/501100011033), under grant PID2019-106715GB-C21. Manuel de León and
Asier López-Gordón recieved support under ``Severo Ochoa Programme for Centres of Excellence in R\&D'' (CEX2019-000904-S), funded by MCIN/AEI/ 10.13039/501100011033. 
Manuel de León also acknowledges the grant EIN2020-112197, funded by AEI/10.13039/501100011033 and European Union NextGenerationEU/PRTR.
Asier López-Gordón would also like to thank MCIN/AEI/ 10.13039/501100011033 for the predoctoral contract PRE2020-093814.



\section*{Data availability}
Data sharing is not applicable to this article as no new data were created or analyzed in this study.

\printbibliography
\end{document}